\documentclass[a4paper,12pt]{article}
\usepackage{amsmath}
\usepackage{graphicx}
\usepackage{rotating}
\usepackage{wasysym}
\DeclareGraphicsExtensions{.pdf,.png,.jpg}
\usepackage{chicago}
\usepackage{klett}
\usepackage{def_the}

\usepackage{Fig_size_A}
\usepackage{kfig}

\renewcommand{\theequation}{\mbox{\arabic{section}.\arabic{equation}}}
\renewcommand{\thefigure}{\arabic{section}.\arabic{figure}}
\renewcommand{\thetable}{\arabic{section}.\arabic{table}}
\renewcommand{\footnoterule}{\rule{14.8cm}{0.3mm}\vspace{+1.0mm}}
\renewcommand{\baselinestretch}{1.0}
\newtheorem{corollary}[theorem]{Corollary}

\title{}
\author{Eckhard Platen}
\begin{document}
\thispagestyle{empty} \vspace*{1.0cm}

\begin{center}
{\LARGE\bf Information-Theoretic Approach to  \\\vspace{0.5cm} Financial Market Modelling  } 
\end{center}

\vspace*{.5cm}
\begin{center}

{\large \renewcommand{\thefootnote}{\arabic{footnote}} {\bf Eckhard
Platen}\footnote{University of Technology Sydney,
  School of Mathematical and Physical Sciences, and \\Finance Discipline Group}$^{,}$
}
\vspace*{2.5cm}

\today

\end{center}

\begin{minipage}[t]{13cm} The paper treats the financial market as a communication system, using four information-theoretic assumptions to derive an idealized model with only one parameter. State variables are scalar stationary diffusions. The model minimizes the surprisal of the market and the Kullback-Leibler divergence between the benchmark-neutral pricing measure and the real-world probability measure. The state variables, their sums, and the growth optimal portfolio of the stocks evolve as squared radial Ornstein-Uhlenbeck processes in respective activity times.

\end{minipage}
\vspace*{0.5cm}


\vspace*{0.5cm}
\vspace*{0.5cm} 
\vspace*{0.5cm}\\
\noindent{\em Acknowledgements:\/} The author would like to express his gratitude for
\noindent receiving valuable suggestions on the  paper by Yacine Ait-Sahalia, Mark Craddock, Kevin Fergusson, Len Garces, Martino Grasselli, Matheus Grasselli, Juri Hinz, Hardy Hulley, Sebastien Lleo, Erik Schloegl, Thorsten Schmidt, Michael Schmutz, Martin Schweizer,  Stefan Tappe, and three referees.

	\newpage
	\section{Introduction}\label{section.intro}

  The Fundamental Theorem of Asset Pricing, derived by \citeN{DelbaenSc98}, forms the basis of modern financial mathematics, as presented, e.g., by \citeN{Jarrow22}. It states that the widely applied method of risk-neutral pricing is equivalent to the absence of free lunches with vanishing risk (FLVRs). \citeN{PlatenFe25a} showed that the absence of FLVRs in real markets is not supported at a high significance level. This prompts reconsideration of financial market models to allow FLVRs and provide a sound, practical pricing method.\\
  
  The benchmark approach provides a modelling framework where FLVRs can exist; see, e.g., \citeN{Platen06ba} and \citeN{PlatenHe06}. Benchmark-neutral pricing, introduced in \citeN{Platen25a}, uses the stocks' growth optimal portfolio (GOP) as num\'eraire and offers a practical, theory-based pricing approach. Developing related quantitative methods requires a financial market model that captures empirical facts and enables precise contingent claim hedging.\\
   
   The following key stylized empirical facts should be incorporated into any model aiming to accurately reflect financial markets:\\
   
  \noindent 1.	Stock market index volatilities tend to follow processes with stationary probability densities, as demonstrated in \citeN{Engle82}.\\
  
   \noindent 2.	The log-return distributions of stock market indices closely approximate a Student-t distribution with four degrees of freedom across various statistical methodologies, as shown in \citeN{MarkowitzUs96a}, \citeN{MarkowitzUs96b}, \citeN{HurstPl97d}, and \citeN{FergussonPl06dc}.\\ \\
   3.	Stock market indices display the leverage effect: declines in the index are associated with increased volatility, and vice versa; see \citeN{Black76b}, \citeN{AitSahalia12}, and \citeN{ElEuchFuRo18}.\\ \\
   4.	Volatility exhibits clustering and rough sample paths, including occasional spikes that differ from typical diffusion processes over calendar time, as discussed by \citeN{BayerFrGa16}, and \citeN{ElEuchFuRo18}.\\ \\
   5.	Volatility comprises both rapidly and slowly changing components; see \citeN{FouquePaSi99}.\\ \\
   6.	Log-returns of major indices exhibit scaling properties and self-similarity, as evidenced by \citeN{Mandelbrot01a} and \citeN{BreymannLuPl09e}.\\ \\
   7.	Over extended periods, the logarithm of world stock market indices demonstrates approximately linear growth with mean-reverting characteristics, exemplified by the logarithm of the MSCI World total return stock index (MSCI) shown in Figure \ref{FigMSCIMF1}.\\ \\
      A \textquoteleft realistic' financial market model, as referenced in this paper, should provide robust explanations for these observed empirical phenomena.\\

  This paper explores a novel interpretation of the financial market, viewing it as a {\em communication system}. By integrating information-theoretical principles in the modelling, an idealized market model is developed based on four fundamental assumptions. These assumptions serve as the foundation for a rigorous and structured approach to understanding market dynamics. To ensure analytical tractability, the state variables within the market are represented as scalar stationary diffusions. This choice enables a manageable yet robust framework for analyzing the evolution and behaviour of these variables over time. \\
  
  	The existence of the {\em growth optimal portfolio} (GOP) is assumed, providing a benchmark for portfolio performance. The   GOP of a market  is   interchangeably  called the Kelly portfolio, expected logarithmic utility-maximizing portfolio, or num\'eraire portfolio; 
  	see, e.g.,  \citeN{Kelly56}, \citeN{Merton71}, \citeN{Long90}, \citeN{Becherer01}, \citeN{Platen06ba}, \citeN{KaratzasKa07}, \citeN{HulleySc10}, and \citeN{MacLeanThZi11}.  \citeN{Kelly56} identified the GOP in gambling using Shannon's information theory, marking the discovery of the GOP. \\
  	
  		The expected information content, specifically the surprisal of the stationary densities of the state variables, is minimized. This approach seeks the most efficient encoding of market information with the release of the minimal possible expected information content through traded prices.\\
 The benchmark-neutral pricing measure is in \citeN{Platen25a} proposed to replace the putative risk-neutral pricing measure. The Kullback-Leibler divergence between the benchmark-neutral pricing measure and the real-world probability measure is also minimized. This ensures that the pricing density remains as close as possible to observed market probabilities, maintaining both realism and theoretical soundness when pricing contingent claims. \\
  
  The framework that emerges from these information-minimization criteria yields a distinct idealized financial market model, which is called {\em minimal market model} (MMM). It extends the similarly named model proposed in \citeN{Platen01a} for the dynamics of a GOP. The paper will show, under the MMM, the state variables, their aggregated sums, and the GOP of the stocks all evolve according to squared radial Ornstein-Uhlenbeck (SROU) processes as studied in \citeN{GoingYo03}, each within their respective activity times. This result encapsulates the dynamic behaviour of both individual and collective market components within an information-theoretic modelling context.\\
  
 The properties of the proposed model will be derived endogenously. The model will be parsimonious with only one parameter  and a minimum number of driving Brownian motions.  It will facilitate the design of extremely accurate  hedging methods. The modelling of the complex stochastic dynamical system represented by the financial market is significantly simplified by modelling its state variables, the {\em normalized factors},  as independent scalar diffusion processes with stationary probability densities.\\ Since stocks are securities that are driven by several sources of randomness, it is not ideal to employ these as the  stochastic basis of a financial market model. Instead, the paper models stocks as portfolios of independent auxiliary securities, the {\em factors}, which   form the theoretical stochastic basis of the market. A factor is assumed to represent a portfolio of stocks. It is modelled as the product of an exponential function of time and an  independent state variable, the {\em normalized factor}.  By construction, the volatilities of the normalized factors are stationary processes and equal to those of the respective factors.  All securities are denominated in units of the savings account, which is added as a security to the factors to form the {\em extended  market} so that the interest rate does not need to be modelled.

		 The  extended market is assumed to permit continuous trading,   instantaneous investing and borrowing,   short sales with full use of proceeds,    infinitely divisible securities, and      no transaction costs. 
	 For the extended market, the modelling challenge  is reduced to the question,  which independent stationary scalar diffusions  should model the normalized factors. To answer  this question,  
	 the  paper optimizes 
	 suitable Lagrangians that capture mathematically the  principle driving the market's dynamics, which is the minimization of expected information content, the, so  called,  {\em surprisal}. 
	  To capture this principle, the 
	    paper interprets the market as a  communication system in the sense of \citeN{Shannon48}, where the expected information content is minimized.
	    This  identifies the optimal stationary probability densities of the normalized factors. \\

	    The fact that information-theoretical concepts are particularly pertinent to the field of finance can be explained as follows: The trades reveal with their prices information. The average speed of the information flow is measured by the average trading intensity, which the paper models by the {\em market activity}. The integrated market activity is called the {\em market time}. At certain periods, the market moves rather fast and at other periods more slowly. One notices that the market activity represents a  fast moving process because particular information can trigger enormous trading activity in a very short time.  Important is the fact that the market participants observe the market evolution in calendar time. When observing the volatility of a stock market index  in calendar time during a period of high market activity, one     observes   volatility spikes, which a diffusion that is evolving in calendar time cannot generate. The paper assumes  that the market is moving in a common  market time that moves faster at major market movements. However, the state variables, which are the normalized factors,  move in market time as scalar diffusions. This separates the modelling of the fast moving market activity from that of the  in market time as diffusions evolving  normalized factors. The market activity or speed of the information flow can be measured by the average trading intensity. The modelling of the market activity is beyond the scope of the current paper and suggested to be modelled similarly as proposed in \citeN{PlatenRe19}. This leads  to the first fundamental assumption of the  paper:\\
	 
	  \noindent {\em \bf A1}:
	 {\em The  normalized factors evolve as independent, stationary scalar diffusions  in  market time, which is the integrated  average trading intensity. }\\

	  The existence of the GOP can be interpreted as a  {\em no-arbitrage condition}   because \citeN{KaratzasKa07} 
	 have shown that  the existence of the  GOP is equivalent to their {\em No Unbounded Profit with Bounded Risk} (NUPBR)  condition. This no-arbitrage condition is weaker than the no {\em  Free Lunch with Vanishing Risk} (FLVR)  condition of \citeN{DelbaenSc98}. The assumption about the existence of the   GOP   is an extremely weak condition. When violated, the     respective growth rate maximization would have no solution and the candidate for the GOP would reach infinite values in finite time, which is not a financial market that the current paper aims to model.
	  This leads  to the  second fundamental assumption of the  paper:\\
	 
	 \noindent {\em \bf A2}:
	{\em The market's GOP and a
		 unique, strong solution of the system of SDEs, characterizing the market's normalized factor dynamics,    exist.}\\

By maximizing the instantaneous growth rate of portfolios of factors when forming the GOP of factors, which is called the {\em benchmark}, a first Lagrangian comes into play. The Assumption {\bf{A2}} requires the GOP of the entire market to exist, which is called {\em num\'eraire portfolio} (NP) and was named by \citeN{Long90}. By maximizing the growth rate of all admissible portfolios in the market, a second Lagrangian emerges in the modelling. \\ 

The trades disclose  price information. The stationary joint probability density of the normalized factors quantifies the probabilistic nature of the prices.   When the surprisal, as defined in \citeN{Kullback59},  is minimized, the probability densities are so that all price-relevant expected information content that could be removed by the buyer and seller in a trade has been removed  under the given constraints.  This leads to the third Lagrangian  and the third fundamental assumption of the  paper: \\ 
	  
	  \noindent {\em \bf A3}:
	  {\em The surprisal, which is the expected information content of the stationary joint probability density of the normalized factors,  is minimized.}\\

	      The paper calls an extended market  a {\em  surprisal-minimized} market when the above three  assumptions are satisfied.   	       
	         As an alternative and practicable pricing method to risk-neutral pricing,  {\em benchmark-neutral (BN) pricing} has been proposed  in \citeN{Platen25a}, which employs as its num\'eraire the benchmark that equals the GOP of the factors.     Buyers and  sellers of contingent claims have  an  interest to minimize the increase of the expected information content caused by  the pricing  density they use. 
	          The derivative of the {\em Kullback-Leibler divergence},  as defined in  \citeN{Kullback59}, of the   BN pricing measure from the real-world probability measure quantifies the increase of expected information content as the result of BN pricing. Under the first three assumptions, the latter will be shown to be the only practicable pricing method that is applicable for all admissible contingent claims. This  identifies  the fourth Lagrangian   and leads to the formulation of the  fourth fundamental assumption of the  paper: \\
	       
	       \noindent {\em \bf A4}:
	       {\em The Kullback-Leibler divergence of the benchmark-neutral pricing measure   from the  real-world probability measure is minimized.}\\
	       
	        We call an extended market   an {\em  information-minimized} market when the four assumptions are satisfied by its dynamics. 
	      The  paper will show that  the above four assumptions lead to an idealized financial market model that is \textquoteleft realistic' in the sense that it generates the empirical facts previously listed.  As indicated earlier, the derived idealized information-minimized financial market model is called minimal market model (MMM).  The MMM has  only one constant  parameter when viewed in market time, which is the {\em net risk-adjusted return}. This parameter is not  needed when applying benchmark-neutral pricing under the MMM.  The observable market time is the only input that is required for the pricing and hedging of a wide range of contingent claims expressed in terms of the market time.\\  
	      Under   the BN pricing measure, the  factors, sums of factors and the benchmark  follow  squared Bessel processes that evolve in respective observable activity times. Squared Bessel processes are special squared radial Ornstein-Uhlenbeck (SROU) processes, which are studied, e.g., in \citeN{RevuzYo99} and \citeN{GoingYo03}. They are generalizations of the CIR process, which was popularized in finance in   \citeN{CoxInRo85}. These processes are highly tractable with many explicit formulas for their functionals; see, e.g.,  \citeN{BaldeauxPl13}.  
	    	\\

	     The paper is organized as follows:  Section 2  introduces  the  extended market of factors. Section 3 minimizes the surprisal of the stationary densities of the normalized factors.  Section 4 minimizes the Kullback-Leibler divergence of the BN pricing measure from the real-world probability measure, which  reveals the  MMM. 
	      Five appendices prove the  theorems of the paper.

	            \section{ 
	         	Market of Factors}\setcA \setcB \label{Section3}
	        
	         \subsection{Factors}
	         
	         Under Assumption {\bf{A1}}, the nondecreasing {\em market time} $\tau_t$ evolves    with respect to the calendar time  $t\in[t_0,\infty)$. 
	        The  modelling is performed on a filtered probability space $(\Omega,\mathcal{F},\underline{\cal{F}},P)$ in market time, satisfying the usual conditions; see, e.g., 
	         \citeN{KaratzasSh98}.  The filtration $\underline{\cal{F}}$ $=(\mathcal{F}_\tau)_{\tau \in [\tau_{t_0},\infty)}$ models the evolution of all events relevant to the  modelling.\\
	   	        The securities are denominated in units of the savings account $S^0_\tau\equiv1$ and specify  $n\in\{1,2,...\}$ stocks $A^1_\tau,...,A^n_\tau$  as the risky primary security accounts. The stocks   reinvest all dividends or other payments and expenses. The $n$ stocks are assumed to represent self-financing portfolios of  the
	        $n$ independent {\em  factors}  $S^1_\tau,...,S^n_\tau$ at the market time $\tau\in[\tau_{t_0},\infty)$.  The $k$-th factor $S^k_\tau$, $k\in\{1,...,n\}$,   is assumed to be  driven by the $k$-th Brownian motion $W^k_\tau$.  The  $n$  independent driving Brownian motions $W^1_\tau,...,W^n_\tau$  evolve in  market time $\tau\in[\tau_{t_0},\infty)$ under the real-world probability measure $P$, where the filtration generated by the Brownian motions is a subfiltration of the given filtration $\underline{\cal{F}}$.\\
	        
	        The factors are independent, nonnegative auxiliary securities that form the stochastic basis of the market. We assume that each factor can be formed as a self-financing portfolio of the stocks, that is, we have the stochastic differential equation (SDE)
	        \begin{equation}
	        	\frac{	dS^j_\tau}{S^j_\tau}=\sum_{k=1}^{n}p^{k,j}_\tau \frac{	dA^k_\tau}{A^k_\tau}
	        \end{equation}
	        with initial value $S^j_{\tau_{t_0}}>0$, where the weights $p^{j,1}_\tau,...,p^{j,n}_\tau$ form predictable, square integrable processes for $\tau\in[\tau_{t_0},\infty)$  and  $j\in\{1,...,n\}$.\\
	        
	        Under Assumption {\bf{A2}}, the  market formed by the $n$ factors has a {\em growth optimal portfolio} (GOP), which we call the {\em benchmark} and denote  by $S^*_\tau$ at the market time $\tau$. Since the set of portfolios that can be formed by the factors and those that can be formed by the stocks are assumed to be the same, it follows by Theorem 3.1 in \citeN{FilipovicPl09} that the benchmark is  also the GOP of the market formed by the stocks. The market formed by the factors has, by construction, no locally riskfree portfolio (LRP), which is a portfolio with zero volatility. \\
	         Theorem 3.1 in \citeN{FilipovicPl09} reveals the general structure of a financial market with continuous securities that has a GOP and no LRP, which determines the structure of the market of factors.
	        The most striking structural property of this market is the existence of a unique   {\em net risk-adjusted return} $\hat\lambda_\tau$, which emerges as the Lagrange multiplier of the growth rate maximization that identifies  the benchmark $S^*_\tau$. 
	        The net risk-adjusted return $\hat\lambda_\tau$ is assumed to represent an   adapted,  square integrable process.  \\
	        We assume for  $k\in\{1,...,n\}$ the   {\em $k$-th  factor process} 
	        $S^k=\{ S^k_\tau,\tau\in[\tau_{t_0},\infty) \}$  to satisfy  the SDE \begin{equation}\label{hatB11j21}
	        	\frac{d   S^k_\tau}{ S^k_\tau}
	        	=\hat\lambda_\tau d\tau+\beta^k_\tau(\beta^k_\tau\omega^k d\tau+dW^k_\tau) 
	        \end{equation}
	        for $\tau\in[\tau_{t_0},\infty)$ with  strictly positive $\mathcal{F}_{\tau_{t_0}}$-measurable, random {\em initial $k$-th factor value} $ S^k_{\tau_{t_0}}>0$. The value of a factor  may approach zero, at which instance it is assumed to be instantaneously reflected.  For each independent driving Brownian motion $W^k_\tau$, $k\in\{1,...,n\}$, the respective {\em $k$-th factor volatility }   is denoted by $\beta^k_\tau$ and the respective {\em risk premium parameter} by $\omega^k>0$.  The volatility with respect to market time is  assumed to represent a flexible, continuous, strictly positive,  adapted, square integrable   process in market time. The expected rate of return $\hat\lambda_\tau+(\beta^k_\tau)^2\omega^k$ is  flexible because it is formed by the sum of the net risk-adjusted return and the product of the flexible $k$-th risk premium parameter  and the squared $k$-th factor volatility. 
	        \\
	      
	         \subsection{Benchmark}

	         Let us denote by $\beta_\tau$ the diagonal matrix with the factor volatilities at its diagonal and zeros at all off-diagonal elements.  Furthermore, we denote by ${\bf \omega}=(\omega^1,...,\omega^n)^\top$ the vector of the risk premium parameters. This allows us to write the SDE for the vector of  factors ${\bf{S}}_\tau=(S^1_\tau,...,S^n_\tau)^\top$ in the form \begin{equation} \label{e.2.14}
	         	\frac{d{\bf{S}}_\tau}{{\bf{S}}_\tau}=\hat\lambda_\tau{\bf{1}} d\tau +\beta_\tau(\beta_\tau\omega d\tau+
	         	d{\bf{W}}_\tau)
	         \end{equation}
	         with the vector of strictly positive {\em initial values}  ${\bf{S}}_{\tau_{t_0}}$ and  the vector process ${\bf{W}}=\{ {\bf{W}}_\tau =(W^1_\tau,\dots,W^{n}_\tau)^\top,\tau \in [\tau_{t_0},\infty)\}$ of the $n$   independent  driving  Brownian motions. 
	         Here we write $\frac{d{\bf{S}}_\tau}{{\bf{S}}_\tau}$ for the $n$-vector of stochastic differentials $(\frac{dS^1_\tau}{S^1_\tau},...,\frac{dS^{n}_\tau}{S^{n}_\tau})^\top$.
	         By application of the It\^o formula one obtains for a self-financing portfolio  $S^{\pi}_\tau$ of factors with  weight vector   ${  \pi^{}_\tau}=(\pi^{1}_\tau,..., \pi^{n}_\tau)^\top$ the SDE
	         \begin{equation} \label{e.2.10}
	         	\frac{dS^{\pi}_\tau}{S^{\pi}_\tau}= \pi_\tau^\top\frac{d{\bf{S}}_\tau}{{\bf{S}}_\tau}
	         \end{equation}
	         and the {\em instantaneous growth rate}
	         \begin{equation}
	         	g^{\pi}_\tau=\hat\lambda_\tau+({\pi}^{}_\tau)^\top \beta_\tau\beta_\tau\left(\omega
	         	-\frac{1}{2}\pi^{}_\tau\right)
	         \end{equation}
	         as the drift of the SDE of its logarithm  for $\tau \in [\tau_{t_0},\infty)$. The benchmark $S^*_\tau$ is the portfolio that maximizes this growth rate, which involves the first Lagrangian and is defined as follows:
	         \begin{definition}\label{GOP}
	         	The benchmark  $S^{*}_\tau$   is the strictly positive  portfolio of factors with maximum instantaneous growth rate $g^{\pi}_\tau$ and initial value $S^{*}_{\tau_{t_0}}>0$, where its  weight vector   ${  \pi^{*}_\tau}=(\pi^{*,1}_\tau,..., \pi^{*,n}_\tau)^\top$
	         	is a solution of the well-posed $n$-dimensional constrained quadratic maximization problem
	         	\begin{equation}\label{quop3}
	         		\max \left\{g^{\pi}_\tau
	         		|\pi^{}_\tau \in {\bf R}^{n},  \pi_\tau^\top  {\bf 1}=1\right\},
	         	\end{equation} 
	         	for all $\tau \in [\tau_{t_0},\infty)$.
	         \end{definition}
	         Note that the above portfolios of factors do not invest in the LRP, which is the savings account.  The following properties of the benchmark are derived in  Appendix A:   
	         \begin{theorem}\label{stockGOP}
	         	For a market of factors and $\tau\in[\tau_{t_0},\infty)$,
	         	the sum of the risk premium parameters is conserved and equals the constant $1$, that is,
	         	\begin{equation}\label{Nt=1}
	         		\omega^\top {\bf{1}}
	         		=1.\end{equation}   When constructing the benchmark, the vector of weights assigned to the factors corresponds to the vector of risk premium parameters 
	         	\begin{equation}\label{pi1}
	         		\pi^{*}_\tau=\omega
	         	\end{equation} 
	         	and the benchmark satisfies the SDE  
	         	\begin{equation} \label{e.4.122a}
	         		\frac{dS^{*}_\tau}{S^{*}_\tau}= \hat\lambda_\tau d\tau +(\beta_\tau\omega)^\top(\beta_\tau\omega d\tau+d{\bf {W}}_\tau)
	         	\end{equation}
	         	for	$\tau \in [\tau_{t_0},\infty)$  with $S^*_{\tau_{t_0}}>0$. 
	         \end{theorem}
	         By Equation \eqref{Nt=1}, the above theorem reveals the important property   that the sum of the risk premium parameters is a conserved quantity and equals $1$. For the case when the risk premium parameters would be predictable, square integrable processes, the respective proof of the above theorem does still work,  the risk premium factors would still  add up to one, and these would still represent the weights for the benchmark. \\

	        \subsection{Num\'eraire Portfolio}
	        The savings account denominated savings account $S^0_\tau\equiv1$ equals the constant $1$
	        for  $\tau\in[\tau_{t_0},\infty)$. 
	        Let us extend the market of factors by adding the savings account $S^0_\tau\equiv1$ to the factors ${\bf{S}}_\tau=(S^1_\tau,...,S^n_\tau)^\top$. In the {\em extended market} one can form  self-financing portfolios by combining the $n+1$  securities  $S^0_\tau,...,S^n_\tau$.   A positive  self-financing  portfolio $S^{\pi}_\tau$ in the extended market is described by its  initial value $S^{\pi}_{\tau_{t_0}}>0$ and its   risky  security  weights    $\pi_\tau=({\pi}^1_\tau,...,{\pi}^n_\tau)^\top$, where the savings account weight equals
	        \begin{equation}\label{11pit1}
	        	\pi^0_\tau=1-  \pi^\top_\tau {\bf{1}}
	        \end{equation}  for $\tau \in [\tau_{t_0},\infty)$. 
	        The portfolio $S^\pi_\tau$ satisfies  the SDE
	        \begin{equation} \label{e.2.100}
	        	\frac{dS^{\pi}_\tau}{S^{\pi}_\tau}=%\pi^0_\tau \frac{dS^0_\tau}{S^0_\tau}+ 
	        	\pi_\tau^\top\frac{d{\bf{S}}_\tau}{{\bf{S}}_\tau},
	        \end{equation}
	        and its instantaneous  growth rate, the drift of its logarithm, is of the form
	        \begin{equation}
	        	g^{\pi}_\tau=
	        	\pi^\top_\tau {\bf{1}}\hat\lambda_\tau+({\pi}^{}_\tau)^\top \beta_\tau\beta_\tau(\omega
	        	-\frac{1}{2}\pi^{}_\tau)
	        \end{equation}
	        for $\tau \in [\tau_{t_0},\infty)$.  \\
	        
	        \citeN{Long90} introduced the num\'eraire portfolio (NP) as a general market notion for the GOP of the extended market, showing that risk-neutral prices can be derived using the NP as numeraire and the real-world probability measure $P$ as pricing measure, which requires an equivalent risk-neutral measure. 
	         The benchmark approach emphasizes that only the existence of the NP is needed to compute minimal prices for replicable contingent claims.  \citeN{KaratzasKa07} showed that the existence of the NP is equivalent to their NUPBR condition. In analogy to Definition \ref{GOP}, and assuming {\bf{A 2}}, the second Lagrangian is then relevant, leading to the following concept:
	        \begin{definition}\label{GOP1}
	        	For the extended market, the NP $S^{**}_\tau$   is the positive  portfolio of the savings account and the factors with maximum instantaneous growth rate $g^{\pi^{**}}_\tau$ and initial value $S^{**}_{\tau_{t_0}}>0$, where its  weight vector   $  (\pi^{**,0}_\tau,\pi^{**,1}_\tau,...,\\ \pi^{**,n}_\tau)^\top=(\pi^{**,0}_\tau,{\pi^{**}_\tau}^\top)$
	        	is a solution of the well-posed $(n+1)$-dimensional constrained quadratic maximization problem
	        	\begin{equation}\label{quop3}
	        		\max \left\{g^{\pi}_\tau 
	        		|(\pi^0_\tau,\pi^\top_\tau)^\top \in {\bf R}^{n+1},\pi^0_\tau+ \pi_\tau^\top  {\bf 1}=1\right\},
	        	\end{equation} 
	        	for all $\tau \in [\tau_{t_0},\infty)$.
	        \end{definition} Under the   Assumption {\bf{A2}},  the existence of the NP is secured and the following result  derived in Appendix B:
	        \begin{theorem}\label{GOPentire}
	        	For the given extended market, the  NP $S^{**}_\tau$ satisfies the SDE
	        	\begin{equation} \label{bare.4.113}
	        		\frac{d S^{**}_\tau}{ S^{**}_\tau}=
	        		{\theta^{**}_\tau}^\top (\theta^{**}_\tau d\tau+d{\bf W}_\tau)
	        	\end{equation}
	        	with initial value $ S^{**}_{\tau_{t_0}}>0$,  {\em  market price of risk vector}
	        	\begin{equation}\label{thetak}
	        		\theta^{**}_\tau=\hat\lambda_\tau \beta_\tau^{-1}{\bf{1}}+\beta_\tau\omega, 
	        	\end{equation}
	        the	NP 
	        	weights
	        	\begin{equation}
	        		\pi^{**}_\tau=( \pi^{**,1}_\tau,..., \pi^{**,n}_\tau)^\top=\hat\lambda_\tau \beta_\tau^{-2}{\bf{1}}+\omega,
	        	\end{equation} to be invested in the factors $S^1_\tau,...,S^n_\tau$, and the weight
	        	\begin{equation}\label{rt3}
	        		\pi^{**,0}_\tau
	        		=-\hat\lambda_\tau{\bf{1}}^\top\beta_\tau^{-2}{\bf{1}} 
	        	\end{equation}
	        	to be invested in the savings account $S^0_\tau$ 
	        	for  	$\tau\in[\tau_{t_0},\infty)$.

	        	\end{theorem}
	        	 
	        The flexibility of factor volatilities and risk premium parameters, along with Theorem 3.1 in \citeN{FilipovicPl09}, ensures a broad set of continuous market dynamics can be described using the extended market model, including stocks formed as portfolios of factors. The theorem remains valid when risk premium parameters are predictable, square-integrable processes.

	           \section{Surprisal-Minimized  Market }\label{Section4} \setcA \setcB
	           \subsection{Normalized Factors} 
	           By applying the Assumption {\bf{A2}}, we  characterize below the state variables, the normalized factors, as independent, stationary, scalar diffusion processes that evolve in respective   activity times. 
	          
	           We assume for  $k\in\{1,...,n\}$ the {\em  $k$-th normalized  factor} $Y^k_{\tau^k_\tau}$ to be given as the ratio
	           \begin{equation}\label{Yktayt1}
	           Y^k_{\tau^k_\tau}=\frac{S^k_\tau}{B_\tau e^{\tau^k_\tau}
	           }
	           \end{equation} with {\em basis exponential}
	           \begin{equation}\label{basis}
	           B_\tau=\exp\left\{\int_{\tau_{t_0}}^{\tau}\hat\lambda_sds\right\}
	           \end{equation}
	           for $\tau\in[\tau_{t_0},\infty)$, representing a scalar diffusion with stationary probability density $p^k_\infty$. The stationary process $Y^k_{\tau^k_\tau}$ is assumed to evolve in the   {\em $k$-th activity time}
	           \begin{equation}\label{bartaukt}
	           \tau^k_\tau=\tau_{t_0}+(\tau-\tau_{\tau_{t_0}}) a^k
	           \end{equation} for $\tau\in[\tau_{t_0},\infty)$, where $a^k>0$ denotes the strictly positive {\em $k$-th activity} and  $\tau_{t_0}\in(-\infty,\infty)$ the {\em  initial market time}.  By Equation \eqref{Yktayt1} it follows the  {\em $k$-th initial factor value} \begin{equation}S^k_{\tau_{t_0}}=Y^k_{{\tau_{t_0}}}e^{{\tau_{t_0}}}>0.\end{equation}  By application of the It\^o formula it is straightforward to see that a normalized factor has the same volatility as the  respective factor.
	           For  
	           the stationary $k$-th normalized factor process $Y^k=\{Y^k_{\tau^k}, \tau^k\in[{\tau_{t_0}},\infty)\}$, $k\in\{1,...,n\}$, we denote 
	           by $p^k_\infty$ its stationary probability density, 
	           where  \begin{equation}{\bf{E}}^{p^k_\infty}\left(f(Y^k_{\tau^k_\tau})\right)=\int_{0}^{\infty}f(y)p^k_\infty (y)dy\end{equation} denotes the {\em average}  with respect to  this density.  \\
	          Four constraints for the $k$-th normalized factor dynamics emerge as follows:\\ First, the average level of a normalized factor is a crucial quantity because it determines the average level of the value of the factor at a given market time, which models an average price.
	           Since there would exist some ambiguity concerning the average level of the $k$-th normalized factor,  we parametrize the arithmetic average of $Y^k$,  with a view towards  the proof of Theorem \ref{informationtheorem1}, so that
	           \begin{equation}\label{omega}
	           4\omega^k= {\bf{E}}^{p^k_\infty}\left(Y^k_{\tau^k_\tau}\right)>0
	           \end{equation}
	           for $\tau\in[\tau_{t_0},\infty)$.\\ Second, under the benchmark approach, the logarithm of portfolios and, in particular, the logarithmic average of the portfolio with the largest instantaneous growth rate, the GOP, are of central importance.  Therefore, we parametrize the logarithmic average of $Y^k$ by the flexible constant
	           \begin{equation} \label{barstalogmean*}
	           {\bf{E}}^{p^k_\infty}(\ln(Y^k_{\tau^k_\tau}))=\zeta^k\in(-\infty,\infty) \end{equation} for $\tau\in[\tau_{t_0},\infty)$.\\
	          
	         Third, by Theorem \ref{stockGOP} and \eqref{Nt=1}, we have  the constraint 
	           \begin{equation}
	           \sum_{k=1}^{n}\omega^k=1.
	           \end{equation}
	           
	           Fourth, with a view towards the proof of Theorem \ref{informationtheorem}, we introduce for the activities the  constraint
	           \begin{equation}\label{averactivity}
	           \sum_{k=1}^{n}\omega^k\sqrt{\frac{1}{a^k}}=1.
	           \end{equation}
	           	           \\
	           
	           It is our goal to specify the $k$-th normalized factor as a flexible, stationary scalar diffusion under the above described constraints. We achieve this by setting the $k$-th factor volatility equal to the function
	           \begin{equation}\label{beta}
	           \beta^k_\tau=\sqrt{\frac{4a^k}{\phi^k(Y^k_{\tau^k_\tau})}},
	           \end{equation} 
	           where the {\em $k$-th volatility function}  $\phi^k(.)$ is  a   flexible, infinitely often continuously differentiable, strictly positive function of the value of the $k$-th normalized factor.  According to Assumption {\bf{A2}},  when assuming stationary normalized factors, it is  required that the NP $S^{**}_\tau$ exists  and that  the $k$-th normalized factor in the $k$-th activity time is modelled by a unique, independent scalar diffusion process; see, e.g., Section 7.7 in \citeN{PlatenHe06}. By the equations \eqref{Yktayt1}, \eqref{hatB11j21}, \eqref{bartaukt}, and an application of the It\^o formula, the $k$-th normalized factor $Y^k_{\tau^k_\tau}$ satisfies  the SDE 
	           \begin{equation}\label{bardYtau}
	           d Y^k_{\tau^k_\tau}=Y^k_{\tau^k_\tau}\left(4\phi^k(Y^k_{\tau^k_\tau})^{-1}\omega^k-1\right)a^k d\tau+Y^k_{\tau^k_\tau}\phi^k(Y^k_{\tau^k_\tau})^{-\frac{1}{2}}\sqrt{4a^k}d W^k_{\tau}
	           \end{equation}
	           of a scalar diffusion with instantaneous reflection at zero for $\tau\in[\tau_{t_0},\infty)$.  In the above SDE, the diffusion coefficient of the $k$-th normalized factor is  that of a  flexible  scalar diffusion process.\\ The drift coefficient is determined using the Itô formula. Consequently, this SDE encompasses a wide variety of scalar diffusion processes that evolve according to their activity times.
	           For instance, assigning $\phi^k(y)=\bar y$ constant models factors with constant volatilities, as described by \citeN{BlackSc73}. If $ \phi^k(y)$ is set as a power function, such as  $ \phi^k(y)=y^x$, the resulting dynamics exhibit constant elasticity of variance–type stochastic volatility, following the work of \citeN{Cox75}. Choosing $\phi^k(y)=y$ (the identity function) produces 3/2-type stochastic volatility models, as noted by \citeN{Heston97} and \citeN{Platen97d}. If $\phi^k(y)=y^{-1}$, the model coincides with the \citeN{Heston93} volatility specification, while $\phi^k(y)=y^{-2}$ corresponds to the volatility of the Bachelier model (\citeN{Bachelier00}). By allowing $\phi^k(.)$ to be any suitable general function of the normalized factor, one obtains a broad family of local volatility function models; see, for example, \citeN{BreedenLi78}, \citeN{DermanKa94}, and \citeN{Dupire94}.
	           Considered in calendar time, the class of stochastic volatility models becomes even broader, since the market time remains highly adaptable and can be defined in numerous ways, with its derivative serving as a subordinator; see \citeN{Bochner55}, \citeN{Clark73}, and \citeN{HaganKuLeWo02}. The Assumption {\bf{A1}} allows sufficient flexibility to choose the market time as a nondecreasing stochastic process, potentially extending beyond the semimartingale framework.\\

	           \begin{definition}We call 
	           the above described market, formed by  the  savings account $S^{0}_\tau$ and the $n$ factors $S^1_\tau,...,S^n_\tau$ as 
	           primary security accounts, a {\em stationary market},  where  the net risk-adjusted return $\hat\lambda_\tau$ represents an %independent,
	           adapted, square integrable process and the normalized factors stationary, independent, scalar diffusions in market time.
	           \end{definition}
	           This does not mean that in a stationary market the factors or resulting stocks have to be stationary processes. Only the volatilities of the normalized factors are modelled in market time as processes with stationary probability densities. The $k$-stationary probability density $p^k_\infty(y)$, $y \in(0,\infty)$, of the $k$-th normalized factor is the solution of a respective  ordinary differential equation, the Fokker-Planck equation, as shown, e.g., in Section 4.5 in \citeN{PlatenHe06}.

	           \subsection{Surprisal Minimization}\label{selfinformationminimization}
	           
The identification of stationary dynamics in normalized factors depends on the principle driving financial market behaviour. For a stationary market, this principle is formulated through a Lagrangian based on stationary probability densities, which must be minimized to find the optimal joint density. Shannon's information theory, as described in \citeN{Shannon48}, highlights information content as a key concept derived from event probabilities, with the {\em surprisal} representing the expected information content. This paper assumes in Assumption {\bf{A3}} that market prices minimize the surprisal of the joint probability density of normalized factors, suggesting prices embed maximum beneficial information on average. Information content for  a continuous variable is measured by the negative logarithm of its probability density, as described in\citeN{Kullback59}, and surprisal is the expected information content. Thus, rather than focusing on expected utilities, returns, or monetary measures, fundamental financial market quantities are in this paper interpreted as information-theoretic quantities. Following Shannon and Kullback, this leads to the proposed notion below.

	           \begin{definition}
	           For a stationary market  
	           with   stationary joint density 
	           \begin{equation}p_{\infty}=\prod_{k=1}^{n}p^k_{\infty}\end{equation} of the normalized factors, the {\em surprisal} $\mathcal{I}(p_\infty)$ of $p_\infty$ is defined	as  its average information content, which is given by the sum of the integrals
	           \begin{equation}\label{Ipq}
	           	\mathcal{I}(p_\infty)=	
	           	-\sum_{k=1}^{n}	\int_{0}^{\infty} p^k_\infty(y)\ln(p^k_\infty(y))dy,
	           \end{equation} 
	           as long as, the above quantities exist.  
	           \end{definition}
	          The stationary market reaches its lowest surprisal when, within the given specific parametrization and set of constraints, the expected information content of the stationary joint density of normalized factors is minimized. This process yields the third Lagrangian in our pursuit of an idealized market model. We now present the following kind of market: 
	           \begin{definition} A stationary   market  
	           is said to be {\em surprisal-minimized} 
	           if	the surprisal  
	           $\mathcal{I}(p_\infty)$ 
	           of the stationary joint density $p_\infty$ of the normalized factors   is minimized for the  given parametrization and constraints.
           \end{definition} 
   The surprisal-minimized market can be seen as 'efficient,' though this differs from efficiency as defined by \citeN{Fama70} and \citeN{GrossmanSt80}. Negative surprisal, mathematically equal to entropy, is important in the derivation of natural laws across disciplines; see \citeN{KosmannSc18}. The next theorem explains the results of minimizing the surprisal of  $p_\infty$, as per Assumption {\bf{A3}}; its proof is provided in Appendix C.
\begin{theorem}\label{informationtheorem1} For $k\in\{1,...,n\}$ and $\tau\in[\tau_{t_0},\infty)$, the dynamics of the $k$-th normalized factor $Y^k_{\tau^k_\tau}$ of a surprisal-minimized   market is that of a  CIR process  of dimension $\frac{4}{n}$, which evolves  in the $k$-th activity time $\tau^k_\tau$ as the unique strong solution of 
the SDE
\begin{equation}\label{dY21}
dY^k_{\tau^k_\tau}=\left(\frac{4}{n} -Y^k_{\tau^k_\tau}\right)a^kd \tau+ \sqrt{  Y^k_{\tau^k_\tau}4a^k} d W^k_\tau,
\end{equation}	
with its initial value $
Y^k_{{\tau_{t_0}}}=S^k_{\tau_{t_0}}e^{-\tau^1_{\tau_{t_0}}}$
distributed according to the surprisal-minimized stationary density $  p^{k}_{\infty}$. The latter  is  a gamma density with $\frac{4}{n}$ degrees of freedom and mean $\frac{4}{n}$. 
\end{theorem}

As requested by Assumption {\bf{A3}}, the $k$-th normalized factor is a stationary scalar diffusion. It  generates a stationary volatility process in market time and has, by construction, the same volatility as the  $k$-th factor.
Therefore, the squared volatility \begin{equation}(\beta^k_\tau)^2=\frac{4a^k}{Y^k_{\tau^k_\tau}}\end{equation} of the $k$-th  factor with respect to the market time $\tau$ is generated by the squared volatility  of the  $k$-th normalized factor. It follows by application of the It\^o formula and using \eqref{beta} the {\em $3/2$ stochastic volatility dynamics} that are characterized by the SDE
\begin{equation}
d (\beta^k_\tau)^2= \left(a^k+\left(1-\frac{1}{n}\right)(\beta^k_\tau)^2\right)(\beta^k_\tau)^2d\tau-\left((\beta^k_\tau)^2\right)^{\frac{3}{2}}dW^k_\tau	\end{equation}
for $\tau\in[\tau_{t_0},\infty)$, as long as $Y^k_{\tau^k_\tau}>0$. A $3/2$ volatility model was  proposed in \citeN{Platen97d}   as the squared volatility model of a GOP. The latter model became in \citeN{Platen01a} named  {\em minimal market model} (MMM); see \citeN{HulleySc10}. 	A $3/2$-type volatility model was independently suggested in \citeN{Heston97} for derivative pricing due to its excellent tractability. The current paper names the entire idealized financial market model MMM that results from applying all four assumptions.

By applying Equation~\eqref{Yktayt1}, \eqref{hatB11j21}, and the It\^o formula, one obtains the factor dynamics as follows:
\begin{corollary}
For a  surprisal-minimized market and $k\in\{1,...,n\}$,   the $k$-th factor can be expressed  as the product
\begin{equation}\label{1Akt2}  S^k_\tau=Y^k_{\tau^k_\tau}B_\tau e^{\tau^k_\tau},
\end{equation}  which satisfies   in the market time $\tau \in[\tau_{t_0},\infty)$ the SDE
\begin{equation}\label{Akt} d   S^k_\tau= \hat\lambda_\tau S^k_\tau d\tau+
\frac{4}{n}B_\tau a^k e^{\tau^k_\tau}d\tau
+\sqrt{ S^k_\tau 4B_\tau a^k e^{\tau^k_\tau}}dW^k_\tau   
\end{equation}
of  a squared radial Ornstein-Uhlenbeck  (SROU)  process with  dimension $\frac{4}{n}$  for $\tau \in[\tau_{t_0},\infty)$.   
\end{corollary}

\subsection{Benchmark Dynamics}\label{sec2.3}
Applying Theorem \ref{stockGOP}, we see that in a market that minimizes surprisal, the benchmark is created by assigning equal weights to all factors. Specifically, at any market time $\tau\in [\tau_{t_0},\infty)$, the weight of the $k$-th factor is $\pi^{*,k}_\tau=\frac{1}{n}$. Since the volatility of a factor can become infinite if that factor reaches zero, we define  $\mathcal{T}\subset [\tau_{t_0},\infty)$  as the set of all market times when every one of the $n$ factors remains strictly positive. The SDEs governing the SROU processes for both the factors and their normalized counterparts have unique strong solutions, and these solutions are still well-defined even at points when a factor hits zero. To prevent complications from infinite volatility, we focus only on market times contained in $\mathcal{T}\subset [\tau_{t_0},\infty)$. The dynamics for the benchmark under these conditions are detailed in Appendix D.

\begin{theorem} \label{stockGOP11}
	Consider a surprisal-minimized  market at a market time $\tau\in\mathcal{T}$.	     	The  benchmark  $S^*_\tau$ represents the equal-weighted factor portfolio and
	satisfies  the SDE
	\BE \label{e.4.1181}
	d S^*_\tau= \hat\lambda_\tau S^*_\tau d\tau+
		S^*_\tau Z_\tau d\tau+S^*_\tau  \sqrt{Z_\tau} dW^*_\tau=\hat\lambda_\tau S^*_\tau d\tau+
	4e^{\tau^*_\tau}a^*_\tau d\tau+\sqrt{S^*_\tau}  \sqrt{4e^{\tau^*_\tau}a^*_\tau} dW^*_\tau, \EE
	with initial value $ S^*_{\tau_{t_0}}>0$ 
	and {\em  squared benchmark volatility} \begin{equation}\label{Zt}
		Z_\tau= \frac{1}{n^2} \sum_{k=1}^{n}\frac{  4a^k}{Y^k_{\tau^k_\tau}},
	\end{equation} 
			where $W^*_\tau$ is  a Brownian motion  with stochastic differential
	\begin{equation}\label{W*t}
		dW^*_\tau= Z_\tau^{-\frac{1}{2}}\frac{1}{n}\sum_{k=1}^{n}\sqrt{\frac{4a^k} {Y^k_{\tau^k_\tau}}}d W^k_\tau
	\end{equation} and initial value $W^*_{\tau_{t_0}}=0$.          
		The {\em normalized benchmark}
	\begin{equation}\label{Y*1}
		Y^*_{\tau^*_\tau}=\frac{ S^*_\tau}{ e^{\tau^*_\tau}}
	\end{equation}  
	follows an SROU process of dimension four   
	that is	evolving in the  {\em benchmark time} 
	\begin{equation}
		\tau^*_\tau=\tau^*_{\tau_{t_0}}+
		\int_{\tau_{t_0}}^{\tau}a^*_sds
	\end{equation}
	and       	 	           	 	
	satisfies the SDE
	\begin{equation}\label{Y^*_t}
		d Y^*_{\tau^*_\tau}=\hat\lambda_\tau Y^*_{\tau^*_\tau}d\tau+\left(4-Y^*_{\tau^*_\tau}\right)a^*_\tau d\tau+2\sqrt{Y^*_{\tau^*_\tau}a^*_\tau}dW^*_\tau\end{equation}
		with initial value 
	\begin{equation}\label{S*initial}
		Y^*_{\tau^*_{\tau_{t_0}}}=S^*_{\tau_{t_0}}e^{-\tau^*_{\tau_{t_0}}}
	\end{equation} and  {\em benchmark activity}          
	\begin{equation}\label{a*Zt11}
		a^*_\tau=\frac{Z_\tau Y^*_{\tau^*_\tau}}{4}.
	\end{equation}
	\end{theorem}
The normalized benchmark is an  SROU process with dimension four in benchmark time $\tau^*_\tau$, always positive and uniquely solving the SDE \eqref{Y^*_t}. Its squared volatility $Z_\tau$ with respect to market time, as given in  \eqref{Zt}, becomes infinite when any normalized factor reaches zero, causing an upward jump in benchmark time. This highlights that benchmark volatility may spike at certain moments; however, this does not imply that the diffusion coefficient itself becomes infinite. It also suggests the market time-integrated growth rate of the benchmark can have upward jumps when a factor hits zero.
\\
The benchmark is built using a dynamic asset allocation approach. In actual trading, reallocations can only occur at specific intervals since trades are discrete, and the SDEs representing both the factors and the benchmark are theoretical models that assume continuous trading. The continuous-time framework  approximates what typically occurs when factor values become extremely small.\\

 Brownian motion and its Gaussian transition probability density effectively describe the limiting behavior in continuous time for a sum of independently moving particles. This principle underlies models such as Bachelier's \citeN{Bachelier00} approach for sums of independent capital units and the Black-Scholes model by \citeN{BlackSc73}, where it is applied to the logarithm of these sums. Empirical evidence, such as the Student-t log-return distribution observed in stock indices, indicates that these models do not sufficiently represent real market behavior. When analyzing the evolution of aggregated independent capital units, it becomes apparent that their trajectory is influenced by the exact count of units present. Within this framework, the SROU process and its noncentral chi-square transition probability density effectively model the limiting continuous-time dynamics of a system comprising independently moving particles that may generate new entities or cease to exist, as detailed in \citeN{Feller71}. This insight offers an  intuitive perspective on the derived surprisal-minimized market dynamics, framing them as the aggregate activity of independent capital units that either proliferate or disappear autonomously.

\section{Information-Minimized Market }\label{Section6} \setcA \setcB

\subsection{Benchmark-Neutral Pricing}\label{BN}
Hedging an inexpensive zero-coupon bond, which pays one unit of the savings account at maturity, necessitates trading both the num\'eraire and the savings account; for further details, refer to \citeN{FergussonPl23} and  \citeN{BaroneadesiPlSa24}. This requirement implies that, in general, it is essential to be able to trade the num\'eraire used for pricing when hedging contingent claims. Subject to Assumption {\bf{A2}}, real-world pricing may be employed with the NP $S^{**}_\tau$ as the num\'eraire and the real-world probability measure $P$ as the pricing measure, as outlined in \citeN{Platen06ba}. Regarding the derived surprisal-minimized dynamics, risk-neutral pricing can result in prices that are considerably higher than necessary, as demonstrated by \citeN{Platen25a}.\\ Within the current financial market, the NP represents a leveraged portfolio that takes a short position in the savings account. Since practical trading occurs at discrete intervals, a proxy for the NP, an extremely volatile portfolio, cannot be assured to remain strictly positive. Additionally, modelling the dynamics of the NP requires knowledge of the net risk-adjusted return $\hat\lambda_\tau$, a drift parameter process that is particularly challenging to estimate. To address these difficulties, \citeN{Platen25a} introduced the benchmark-neutral (BN) pricing method, which utilizes the benchmark $S^*_\tau$ as the num\'eraire and the corresponding BN pricing measure $Q_{S^*}$ as the pricing measure. The BN pricing measure $Q_{S^*}$  is specified by 
the  {\em  Radon-Nikodym derivative} 
\begin{equation}\label{lambda1}
	\Lambda_{Q_{{{S^{*}}}}}(\tau)=\frac{dQ_{ { S^{*}}}}{dP}\Big |_{\mathcal{F}_{\tau}}=\frac{\frac{ { S^{*}_\tau}}{S^{**}_\tau}}{\frac{ { S^{*}_{\tau_{t_0}}}}{S^{**}_{\tau_{t_0}}}}
\end{equation}
for $\tau\in[\tau_{t_0},\infty)$;  see \citeN{GemanElRo95}. \\

Let ${L}_Q^1(\mathcal{F}_{\tau_T})$
denote the set of $Q$-integrable, $\mathcal{F}_{\tau_T}$-measurable random variables in a  filtered probability space $(\Omega,\mathcal{F},\underline{\cal{F}},Q)$. It is shown in \citeN{Platen26a} for a more  general financial market model than the above derived surprisal-minimized market model  that the BN pricing measure $Q_{S^*}$ is an equivalent probability measure and BN pricing determines the minimal possible  prices for replicable contingent claims. More precisely, one can conclude directly from Theorem 3.4 in \citeN{Platen26a} the following facts:
	\begin{corollary}
	\label{ABNPricingformula}
For a surprisal-minimized  market with  $n>2$ and $\tau\in[\tau_{t_0},\infty)$, where  the absolute value of the  net   risk-adjusted return $\hat\lambda_\tau$ is constant and finite,  the following statements hold:\\
	(i) The  BN pricing measure $Q_{S^*}$ is an equivalent  probability measure. \\
	(ii) For a replicable contingent claim  $ H_{\tau_T}$, with  stopping time $\tau_T$ that is bounded as a market time, where $\frac{H_{\tau_T}}{S^{**}_{\tau_T}}\in L^1_P(\mathcal{F}_{\tau_T})$, and $\frac{H_{\tau_T}}{S^{*}_{\tau_T}}\in L^1_{Q_{S^*}}(\mathcal{F}_{\tau_T})$,
	its minimal possible  price $H_\tau$ at the market time $\tau\in[\tau_{t_0},\tau_T]$ equals the BN price 
	provided by  the {\em benchmark-neutral pricing formula} 
	\BE \label{BPF}H_\tau =S^*_\tau {\bf E}^{Q_{S^*}}\left(\frac{H_{\tau_T}}{S^*_{\tau_T}}|\mathcal{F}_{\tau}\right).
	\EE
		(iii)  For $k\in\{1,...,n\}$, the process $\bar W^{k}=\{\bar W^{k}_\tau, \tau\in[\tau_{t_0},\infty)\}$,
	satisfying the SDE
	\begin{equation}\label{AW0}
		d\bar W^{k}_\tau
		=  \hat\lambda_\tau\sqrt{\frac{S^k_{\tau}}{4a^k_\tau B_\tau e^{\tau^k_\tau} }}d\tau+dW^k_\tau
	\end{equation}
	for $\tau\in[\tau_{t_0},\infty)$ with $\bar W^{k}_{\tau_{t_0}}=0$, is under the BN pricing measure $Q_{S^*}$ a Brownian motion in market time.\\
	(iv) The benchmark satisfies the SDE
	\begin{equation} \label{Ae.4.1192}
		d S^*_\tau= 
			Z_\tau S^*_\tau d\tau+ \sqrt{ Z_\tau  } S^*_\tau d\bar W^*_\tau 
			= 	4 e^{\tau^*_\tau}a^*_\tau d\tau+ \sqrt{ S^*_\tau  4e^{\tau^*_\tau}a^*_\tau  } d\bar W^*_\tau
		\end{equation}
		of a time-transformed squared Bessel process of dimension four,	where  $\bar W^{*}=\{\bar W^{*}_\tau, \tau\in[\tau_{t_0},\infty)\}$ is  under the BN pricing measure $Q_{S^*}$ a Brownian motion in market time,
		satisfying the SDE
		\begin{equation}\label{AW04}
			d\bar W^{*}_\tau
			= Z_\tau^{-\frac{1}{2}}\hat\lambda_\tau d\tau+dW^*_\tau
		\end{equation}
		for $\tau\in[\tau_{t_0},\infty)$ with $Z_\tau$ given in \eqref{Zt} and $\bar W^{*}_{\tau_{t_0}}=0$.
			\end{corollary}
	The BN pricing measure $Q_{S^*}$ substitutes in the  SDEs  of the factors under the real-world probability measure the net risk-adjusted return  $\hat\lambda_\tau$   by zero, and the processes $W^1_\tau,...,W^n_\tau$ by respective processes  $\bar W^1_\tau,...,\bar W^n_\tau$, which are  Brownian motions under the BN pricing measure. This means, the knowledge of the net risk-adjusted return $\hat\lambda_\tau$ is not needed for BN pricing and hedging. The   dimensions of the  SROU processes that form the   normalized factors, factors,  and the  benchmark remain under the BN pricing measure $Q_{S^*}$ the same as under the real-world probability measure $P$. \\

 In a wide class of potential num\'eraires the benchmark has been shown in \citeN{Platen26a} to be the only portfolio that yields as num\'eraire an equivalent pricing measure. In \citeN{SchmutzPlSc25} BN pricing has been generalized for not fully replicable contingent claims by introducing the concept of BN risk-minimization.

\subsection{Information Minimization}	

 We denote for the normalized factors the BN stationary joint density by the product
\begin{equation}q_\infty=\prod_{k=1}^{n}q^k_\infty\end{equation} and the real-world stationary joint density by \begin{equation}p_\infty=\prod_{k=1}^{n}p^k_\infty.\end{equation}  
A trade at the market time $\tau\in[\tau_{t_0},\infty)$ reveals information through its price. Trades are discrete in time and the   trading intensities, which quantify the speed of the flow of  price information,  are modelled by the market activities in the surprisal-minimized   market.   
By using \citeN{Kullback59}, we employ the following notion:

\begin{definition} For    a surprisal-minimized market with BN  stationary joint density $q_\infty$ and real-world stationary  joint density $p_\infty$ of the normalized factors at the market time $\tau\in[\tau_{t_0},\infty)$, the {\em Kullback-Leibler divergence} $I(p_\infty,	q_\infty)$ of  $q_\infty$ from  $p_\infty$  is defined as
	\begin{equation}\label{Ipq1}
		I(p_\infty,	q_\infty)=\frac{d }{d\tau}{\bf{E}}^{p_\infty}\left(-\ln(\Lambda_{Q_{S^*}}(\tau))\right).
	\end{equation}
	. 
\end{definition} 

 The above Kullback-Leibler divergence measures  the increase of expected information content when employing BN pricing. One notices by \eqref{Ipq1} and \eqref{lambda1}  that the Kullback-Leibler divergence equals 
 the average growth rate of the NP  $S^{**}_\tau$ denominated in units of the benchmark $S^*_\tau$. 	\\

The Assumption {\bf{A4}} postulates that the increase of expected information content from BN pricing is minimized in the market, which means that  the Kullback-Leibler divergence $	I(p_\infty,	q_\infty)$ is minimized. This minimization involves our fourth Lagrangian. % required by the Noether method.
\begin{definition}
	We call a surprisal-minimized market {\em  information-minimized} when  the above  Kullback-Leibler divergence of $q_\infty$ from $p_\infty$ is minimized.
\end{definition}       	   	           	    The results of the above minimization   are described in the following theorem, where the proof is provided in Appendix E: 

\begin{theorem}\label{informationtheorem} For an information-minimized market, it follows for $k\in\{1,...,n\}$ that the $k$-th activity equals
\begin{equation}\label{act}
	a^k=1,
\end{equation}
and the net risk-adjusted return
\begin{equation}\label{interestrate}
	\hat\lambda_\tau	=\hat\lambda
\end{equation}  equals a  constant,	           	    
where the divergence of $q_\infty$ from $p_\infty$ amounts to
\begin{equation}\label{infoI}
	I(p_\infty,	q_\infty)=\frac{\hat\lambda^2}{2}
\end{equation}	           	    for $\tau\in[\tau_{t_0},\infty)$. 

\end{theorem}

For an information-minimized market the activities equal $1$ and the net risk-adjusted return $\hat\lambda_\tau=\hat\lambda$ is a constant. As indicated earlier we name the resulting model as follows:
\begin{definition}
	The idealized financial market model that results from the application of the four assumptions {\bf{A1}}, {\bf{A2}}, {\bf{A3}},and {\bf{A4}}, yielding an information-minimized market, is called {\em minimal market model} (MMM).
	\end{definition}

Under the MMM and for  $k\in\{1,...,n\}$, the $k$-th factor 
\begin{equation}\label{Akt2}  S^k_{\tau}=Y^k_{\tau}B_\tau e^{\tau},
\end{equation}   has by \eqref{hatB11j21},  \eqref{Yktayt1}, and application of the It\^o formula, the following properties:

\begin{corollary}\label{corolvarphiAk}
For   
$k\in\{1,...,n\}$, the  $k$-th factor has under the MMM the dynamics  of  an  SROU  process with {\em dimension} $d_k=\frac{4}{n}$. %,
It 
satifies the SDE
\begin{equation}\label{Akt1} d   S^k_{\tau}= \hat\lambda S^k_\tau d\tau+
\frac{4}{n}B_\tau e^{\tau} d\tau+\sqrt{ S^k_\tau}\sqrt{4B_\tau e^{\tau}}dW^k_\tau   
\end{equation}
for $\tau \in[\tau_{t_0},\infty)$,	 with             	          	  random  initial value $ S^k_{\tau_{t_0}}=Y^k_{\tau_{t_0}}e^{\tau_{t_0}}$, where $Y^k_{\tau_{t_0}}$ has the stationary probability density $p^k_{\tau_{t_0}}=p^k_\infty$ 
as its density.

\end{corollary}

\citeN{CraddockPl04} analyze Lie group symmetries for PDEs governing diffusion dynamics with square-root state variable diffusion coefficients. For the normalized factors in \eqref{Akt1}, Theorem 4.4.3 in \citeN{BaldeauxPl13} identifies a specific Lie group symmetry aligning with these dynamics. This symmetry provides an explicit formula for the transition probability density, which is shown to be that of an SROU process; see Equation 5.1.2 in\citeN{BaldeauxPl13}. Therefore, the transition probability density for normalized factor dynamics relates directly to a corresponding Lie group symmetry of its governing PDE.
\\ 
Noether's Theorems in \citeN{Noether18} predict that when a system of partial differential equations for state variables has a Lie group symmetry, a conservation law exists. In an information-minimized market, this conservation law appears as Equation \eqref{interestrate}, which means the net risk-adjusted return remains constant. This makes economic sense and enhances our understanding of how economies work: The parameter $\hat\lambda$, a Lagrange multiplier, can be seen as the average extra net return a firm needs above the interest rate, assuming its business activity is risk-free. If firms earn less than  $\hat\lambda$, their operations would quickly become unprofitable. Conversely, if the average extra net return exceeds  $\hat\lambda$, new competitors would enter the market, driving down returns until they settle at the natural constant  $\hat\lambda$.

\subsection{Additivity Property of Sums of Independent Factors}\label{sumsoffactors}

The sum of independent squared Bessel processes evolving in the same time forms another squared Bessel process, as shown by \citeN{ShigaWa73}. This distinctive additivity arises from the specific partial differential equation governing the transition probability density of the squared Bessel process. These processes are particular types of SROU processes; see \citeN{RevuzYo99}. By applying Corollary \ref{corolvarphiAk}, the following result is derived:

\begin{corollary}\label{sumofatoms}
For  
a set $\mathcal{A}\subseteq \{1,...,n\}$ of indexes of independent factors,  the respective {\em sum of independent factors }
\begin{equation}
S^{\mathcal{A}}_{\tau} =\sum_{k\in\mathcal{A}}  S^k_\tau         \end{equation} satisfies under the MMM the SDE
\begin{equation}\label{Akphi1}
d  S^{\mathcal{A}}_{\tau}
= \hat\lambda S^{\mathcal{A}}_{\tau}d\tau+
d_{\mathcal{A}}B_\tau e^{\tau} d\tau+\sqrt{S^{\mathcal{A}}_{\tau}4 B_\tau e^{\tau}}d W^{\mathcal{A}}_\tau\end{equation}
of an SROU   process 
with net risk-adjusted return $\hat\lambda$,  dimension 
\begin{equation}
d_{\mathcal{A}}=\sum_{k\in\mathcal{A}}\frac{4}{n},
\end{equation}
and	initial value
\begin{equation}
S^{\mathcal{A}}_{\tau_{t_0}} =\sum_{k\in\mathcal{A}}  S^k_{\tau_{t_0}}=\sum_{k\in\mathcal{A}}  Y^k_{\tau_{t_0}}e^{\tau_{t_0}} ,
\end{equation}
where $ W^{\mathcal{A}}_\tau$ is a Brownian motion with stochastic differential
\begin{equation}
d W^{\mathcal{A}}_\tau=\frac{1}{\sqrt{  S^{\mathcal{A}}_{\tau} }}\sum_{k\in\mathcal{A}} \sqrt{ S^k_{\tau}}d W^k_\tau
\end{equation}
for $\tau\in[\tau_{t_0},\infty)$ with initial value $ W^{\mathcal{A}}_{\tau_{t_0}}=0$.
The respective {\em normalized sum of factors} 
\begin{equation}\label{YAP}
Y^{\mathcal{A}}_{\tau^{\mathcal{A}}_\tau}=\frac{	S^{\mathcal{A}}_{\tau}}{B_\tau e^{\tau}}
=\sum_{k \in {\mathcal{A}}}Y^k_{\tau^k_\tau}
\end{equation}
satisfies the SDE
\begin{equation}
\label{Ykappa5}
dY^{\mathcal{A}}_{\tau}	=\left(d_{\mathcal{A}} -Y^{\mathcal{A}}_{\tau}\right)d\tau+\sqrt{4  Y^{\mathcal{A}}_{\tau}} d  W^{\mathcal{A}}_\tau
\end{equation}
and evolves as a CIR process of dimension $d_{\mathcal{A}}$ in the market time 
$\tau\in[\tau_{t_0},\infty)$. 
\end{corollary}

The result demonstrates that factor sums have a unique additive property under the MMM. The sum of all independent factors, called the {\em factor portfolio} (FP)  $S^{FP}_\tau$, is particularly important. Using Corollary \ref{sumofatoms}, the It\^o formula, and Equation \eqref{Nt=1}, we conclude the following: 
\begin{corollary}\label{sumatomAP}
The FP, which is the sum of the factors,
\begin{equation}\label{SAP}
S^{FP}_{\tau}=\sum_{k=1}^{n}  S^k_{\tau}, \end{equation}
follows under the MMM a time-transformed squared Bessel  process of dimension four. The {\em  normalized FP} 
\begin{equation}
Y^{FP}_{\tau^{FP}_\tau}=\sum_{k=1}^{n}Y^k_{\tau^k_\tau}=\frac{S^{FP}_\tau }{B_\tau e^{\tau}}
\end{equation}
represents a CIR process of dimension four in market time with average $4$. 
\end{corollary}
\subsection{Ilustrations}
For the SROU process $S^{FP}_\tau$ evolving in the market time ${\tau_t}$ one obtains via application of the It\^o formula its activity time via the formula 
\begin{equation}
	{\tau_t}=\ln \left(\left[\sqrt{S^{FP}_{\tau_.}}\right]_{t_0,t}+e^{{\tau_{t_0}}}\right)
\end{equation}
for $t\in[t_0,\infty)$; see \citeN{Platen25a}. The quadratic variation $[\sqrt{S^{FP}_{\tau_.}}]_{t_0,t}$ is observable. The only parameter needed to determine the activity time is the initial activity time ${\tau_{t_0}}$. For stock market indices, the activity time typically evolves in an almost linear manner, making it easy to estimate the initial value through linear regression; see \citeN{Platen25a}. This estimation only needs to be done once, after which the initial activity time becomes a known parameter. Once established, it remains unchanged, allowing for straightforward calculation of the current activity time when pricing contingent claims. By setting the maturity of a contingent claim as a fixed activity time value of the underlying security and finely discretizing calendar time, research by
\begin{figure}
	[h!]
	\centering
	\includegraphics[width=14cm, height=6cm]{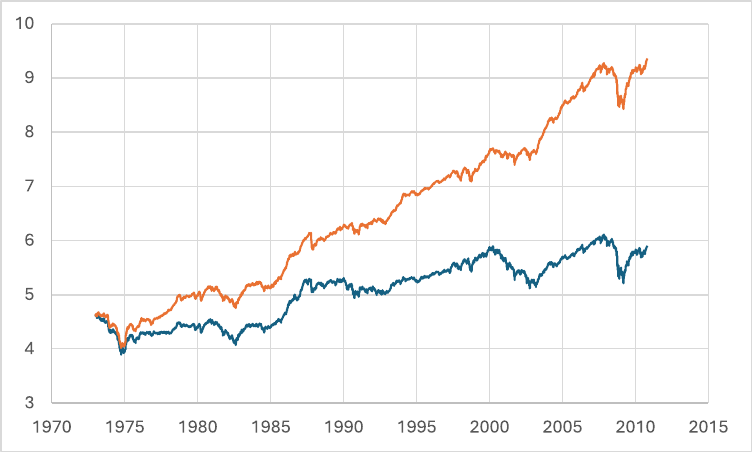}
	\caption{\label{FigMSCIMF1} Logarithms of US  savings account denominated MSCI (blue) and EWI114 (red).}
\end{figure}  
 \citeN{Platen25a} and \citeN{PlatenFe25a} has shown that hedge errors for inexpensive zero-coupon bonds are extremely small. Only noticeable discretization errors arise when activity time moves faster than expected in calendar time, resulting in larger activity time steps. The outstanding hedges achieved are due to the observed activity time capturing all features of the security trajectory not explained by diffusion dynamics modeled in activity time. Most of these features can be accounted for using Assumption {\bf{A1}}, as shown in \citeN{PlatenRe19}.

\begin{figure}
	[h!]
	\centering
	\includegraphics[width=14cm, height=6cm]{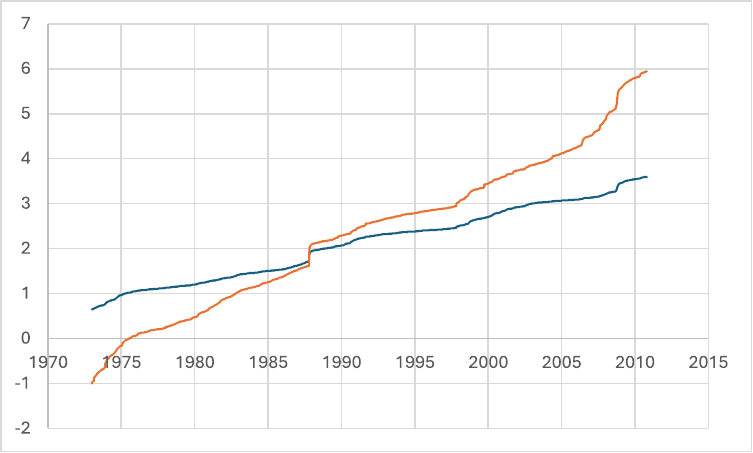}
	\caption{\label{FigMSCIMF2} Market time $\tau_t$ (blue) and benchmark time $\tau^*_{\tau_t}$ (red).}
\end{figure}
For illustration, consider the MSCI World Total Return stock index (MSCI) in US savings account denomination  as proxy for the factor portfolio  $S^{FP}_\tau$ and the EWI114 constructed in \citeN{PlatenRe12an} as proxy for the benchmark $S^*_\tau$. The latter refers to an equally weighted portfolio of 114 global industry subsector indices, which can be considered as factors. Figure \ref{FigMSCIMF1} shows the logarithms of both stock portfolios, revealing their roughly linear progression. Figure \ref{FigMSCIMF2} displays the market time $\tau_t$ and the benchmark time $\tau^*_{\tau_t}$, both following nearly linear paths. As illustrated in \citeN{Platen25a} and \citeN{Platen26a}, for a fixed maturity at the end of the dataset measured in benchmark time, an inexpensive zero-coupon bond, paying out one unit from the savings account, has been priced using BN pricing and hedging, with the EWI114 serving as a proxy. The resulting hedge error, shown in Figure \ref{FigMSCIMF3}, is less than $0.0001$ units of the savings account at maturity, where one unit equals the payout. Notably, hedging inexpensive zero-coupon bonds with very long maturities turns out to be less costly than traditional risk-neutral pricing suggests, an insight highlighted by \citeN{Platen26a} and significant for pension and insurance industries.

\begin{figure}
	[h!]
	\centering
	\includegraphics[width=14cm, height=6cm]{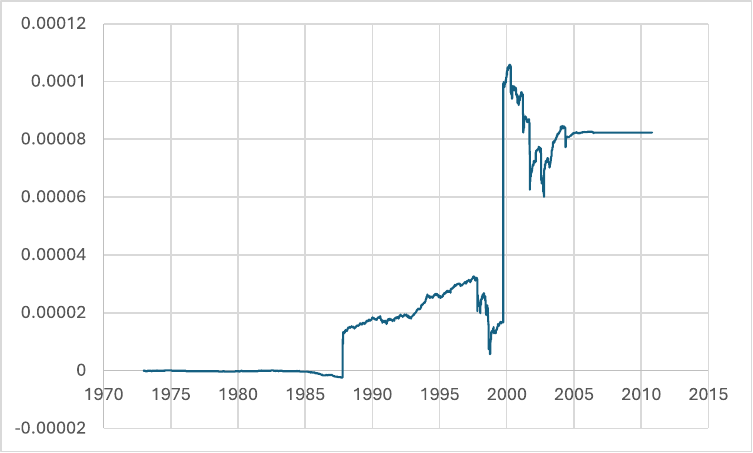}
	\caption{\label{FigMSCIMF3} Hedge error for inexpensive zero-coupon bond.}
\end{figure}

\subsection{Empirical Evidence}
Let us compare predicted properties with empirical data from the real market, as outlined in the introduction: By design, stock market index volatilities are stationary processes; see \citeN{Engle82}. The Student-t density with four degrees of freedom emerges when estimating the log-return density of a security whose squared volatility follows an inverse gamma distribution with four degrees of freedom; see, e.g.,  \citeN{HurstPl97d} or \citeN{PlatenRe08e}. When the MSCI is viewed as the FP, the MMM predicts its log-return density as a Student-t density with four degrees of freedom, which  \citeN{FergussonPl06dc} found to have high significance among numerous competing models. \citeN{BreymannLuPl09e} demonstrated that, across many observation frequencies, the hypothesis that MSCI log-returns follow a Student-t distribution with approximately four degrees of freedom cannot be easily dismissed when valued in major currencies. Similar findings were obtained for   the S\&P500 in  \citeN{MarkowitzUs96a} and  in \citeN{HurstPl97d} for stock indices of various countries.

 Since for an SROU process the volatility  is proportional to the inverse of its square root, the leverage effect, observed by \citeN{Black76a},  is endogeneously built into the MMM dynamics of stock indices and several other securities. The  $3/2$ stochastic volatility model for stock indices of the MMM explains the clustering and heteroscedasticity of observed volatilities. Moreover, the   rough volatility paths observed, e.g., in \citeN{BayerFrGa16}, can be explained, as is shown in \citeN{PlatenRe19}, due to the fact that the market activity evolves in market time but is viewed in calendar time, which generates volatility spikes. The  paper \citeN{PlatenRe19} demonstrates  that the market activity under the MMM is modelling the fast moving component of the stock market index volatility, whereas the inverse of the square root of the normalized stock market index represents its much slower moving component. This aligns with the findings in \citeN{FouquePaSi99}. The logarithm of a stock market index is under the MMM, by construction, approximately linearly evolving in a mean-reverting manner, as visible in Figure \ref{FigMSCIMF1} for the MSCI. As shown in \citeN{PlatenFe25a} and Figure \ref{FigMSCIMF3},  the MMM delivered extremely accurate hedges for thousands of inexpensive extreme maturity bonds.\\
  Under the MMM, benchmark dynamics follow an  SROU  process with dimension four in benchmark time. \citeN{PlatenRe12an} constructed the EWI114 as a proxy for the world stock market benchmark, which shows strong average growth compared to the MSCI, as shown in  Figure \ref{FigMSCIMF1}. Similarly, \citeN{PlatenRe08e} introduced the EWI104; its log-return density is shown to align with a Student-t distribution with four degrees of freedom, consistent with MMM predictions.\\
 \citeN{GoardMa13} found that the 3/2-volatility model of the MMM outperforms other stochastic volatility models when applied to VIX derivatives. \citeN{Mandelbrot01a} observed self-similarity in stock indices, which the MMM also predicts through its modelling of sums of independent factors as squared Bessel processes; see \citeN{RevuzYo99}. Thus, the MMM accounts for key empirical facts and can be considered to represent a \textquoteleft realistic' financial market model.\\

         \section*{Conclusion}
         Based on four assumption, a parsimonious, idealized information-minimized financial market model has been derived. The stock market index and the growth optimal portfolio of the stocks are under this model and under the benchmark-neutral pricing measure   squared Bessel processes that evolve in  respective  activity times.  No parameters are required for the accurate benchmark-neutral pricing and hedging  of a wide range of contingent claims. Only  the  observable market time and the observable benchmark value enter the respective formulas for prices and hedges concerning the benchmark.
          The findings indicate that the financial market can be viewed as a communication system, and  concepts from information theory are particularly pertinent to the field of finance.  Further research  could revisit long standing questions in finance from the perspective of this paper and may lead to interesting answers.

       \appendix

        \textwidth14.5cm \textheight9in
        \topmargin0pt
        \renewcommand{\theequation}{\mbox{A.\arabic{equation}}}
        \renewcommand{\thefigure}{A.\arabic{figure}}
        \renewcommand{\thetable}{A.\arabic{table}}
        \renewcommand{\footnoterule}{\rule{14.8cm}{0.3mm}\vspace{+1.0mm}}
        \renewcommand{\baselinestretch}{1.0}
        \pagestyle{plain}

      \section*{Appendix A: Proof of Theorem \ref{stockGOP}}\setcA \setcB 
     \begin{proof}  
    By applying  Theorem 3.1 in \citeN{FilipovicPl09} for $\tau\in[\tau_{t_0},\infty)$, we obtain from  \eqref{hatB11j21} with the SDE 3.10 	in Theorem 3.1 in \citeN{FilipovicPl09}   
    for the portfolio $S^{\pi}_\tau$ with weight vector $\pi_\tau=(\pi^1_\tau,...,\pi^n_\tau)^\top$ for the holdings in the factors ${\bf{S}}_\tau$ the SDE
    \BE \label{e.4.12}
    \frac{dS^{\pi}_\tau}{S^{\pi}_\tau}=\pi_\tau^\top\frac{d{\bf{S}}_\tau}{{\bf{S}}_\tau}= \hat\lambda_\tau d \tau +
    \sum_{k=1}^{n}\pi^{k}_\tau\beta^k_\tau(\beta^k_\tau \omega^k d \tau+d{ W}^k_\tau). \EE
    By the matrix equation (3.5)   in Theorem 3.1 in \citeN{FilipovicPl09}   it follows  for the $k$-th factor-GOP weight
    \begin{equation}
    	(\beta^k_\tau)^2\pi^{*,k}_\tau+\hat\lambda_\tau= \hat\lambda_\tau+	(\beta^k_\tau)^2\omega^k,
    \end{equation}
    which yields
    \begin{equation}\label{pik}
    	\pi^{*,k}_\tau=\omega^k
    \end{equation}
    for $k\in\{1,...,n\}$ and $\tau\in[\tau_{t_0},\infty)$.   	By 
    applying Equation (3.8) in   \citeN{FilipovicPl09}  we obtain the $k$-th benchmark volatility
    \begin{equation}\label{theta4}
    	\pi^{*,k}_\tau\beta^k_\tau=\omega^k\beta^k_\tau
    \end{equation}
    for  $k\in\{1,...,n\}$ and $\tau\in[\tau_{t_0},\infty)$.  Furthermore, we obtain by Equation (3.4)   in Theorem 3.1 in \citeN{FilipovicPl09} and \eqref{pik} the sum
    \begin{equation}
    	\sum_{k=1}^{n}\omega^k=\sum_{k=1}^{n}\pi^{*,k}_\tau=1,
    \end{equation}
       which proves Equation \eqref{Nt=1}. By \eqref{pik} one obtains the Equation \eqref{pi1} and the SDE \eqref{e.4.122a}, which completes  the proof of Theorem \ref{stockGOP}.\\
    \end{proof}    
    
       	\textwidth14.5cm \textheight9in
              	\topmargin0pt
       	\renewcommand{\theequation}{\mbox{B.\arabic{equation}}}
       	\renewcommand{\thefigure}{B.\arabic{figure}}
       	\renewcommand{\thetable}{B.\arabic{table}}
       	\renewcommand{\footnoterule}{\rule{14.8cm}{0.3mm}\vspace{+1.0mm}}
       	\renewcommand{\baselinestretch}{1.0}
       	\pagestyle{plain}
       	 \section*{Appendix B: Proof of Theorem \ref{GOPentire}}\setcA \setcB
       	
       \begin{proof} 
       	Since the savings account $S^0_\tau\equiv1$ is a traded security in the extended market, it follows from Theorem 3.1 in \citeN{FilipovicPl09} that the  risk-adjusted return of this market is zero. For a savings account denominated portfolio $S^{\pi}_\tau$, which invests with the weight $ \pi^{0}_\tau$ in the savings account $S^0_\tau$ and with the weight $ \pi^k_\tau$ in the $k$-th factor $S^k_\tau$, $k\in\{1,...,n\}$, it follows by  Equation (3.8) and Equation (3.11) in Theorem 3.1 in \citeN{FilipovicPl09}   that this portfolio equals the NP $S^{**}_\tau$  when it satisfies the SDE
       	\begin{equation}
       		\frac{dS^{**}_\tau}{S^{**}_\tau}=\sum_{k=1}^{n}\theta^{**,k}_\tau\left(\theta^{**,k}_\tau d \tau+dW^k_\tau\right)
       	\end{equation}
       	with  \begin{equation}
       		\theta^{**,k}_\tau=\pi^{**,k}_\tau\beta^k_\tau,
       	\end{equation}
       	which proves  Equation  \eqref{bare.4.113}.  By Equation (3.8)  in  Theorem 3.1 in \citeN{FilipovicPl09} we have the  $k$-th market price of risk
       	and by the  matrix equation (3.5)  in  Theorem 3.1 in \citeN{FilipovicPl09}    the equation
       	\begin{equation}
       		(\beta^k_\tau)^2\pi^{**,k}_\tau= \hat\lambda_t+	(\beta^k_\tau)^2\omega^k,
       	\end{equation} which is solved by the optimal $k$-th NP weight
       	\begin{equation}\label{barpi**}
       		\pi^{**,k}_\tau=\frac{ \hat\lambda_\tau}{	(\beta^k_\tau)^2}+\omega^k
       	\end{equation}
       	and yields the $k$-th market price of risk
       	\begin{equation}\label{thetak1}
       		\theta^{**,k}_\tau=\frac{ \hat\lambda_\tau}{	\beta^k_\tau}+\omega^k\beta^k_\tau
       	\end{equation}
       	for $k\in\{1,...,n\}$ and $\tau\in[\tau_{t_0},\infty)$.	 We obtain by \eqref{barpi**} and \eqref{Nt=1} %\eqref{Mt1}
       	the NP weight
       	\begin{equation}
       		\pi^{**,0}_\tau=1-\sum_{k=1}^{n}\pi^{**,k}_\tau=1- \hat\lambda_\tau\sum_{k=1}^{n}(\beta^k_\tau)^{-2}-1= - \hat\lambda^{}_\tau\sum_{k=1}^{n}(\beta^k_\tau)^{-2}
       	\end{equation}
       	to be invested  in the savings account $S^0_\tau$. 
       	This completes the proof of Theorem \ref{GOPentire}.
       \end{proof}

       \textwidth14.5cm \textheight9in
             \topmargin0pt
       \renewcommand{\theequation}{\mbox{C.\arabic{equation}}}
       \renewcommand{\thefigure}{C.\arabic{figure}}
       \renewcommand{\thetable}{C.\arabic{table}}
       \renewcommand{\footnoterule}{\rule{14.8cm}{0.3mm}\vspace{+1.0mm}}
       \renewcommand{\baselinestretch}{1.0}
       \pagestyle{plain}
       \section*{Appendix C: Proof of Theorem \ref{informationtheorem1}}\setcA \setcB
       \begin{proof} 
       	We perform the minimization of the  surprisal
       	\begin{equation}
       		\mathcal{I}(p_{\infty})=	\sum_{k=1}^{n}\mathcal{I}(p^k_{\infty}),
       	\end{equation} of the stationary joint density $p_\infty$ in several steps.\\
       	
       \noindent 1. At first, we derive the stationary densities of the normalized factors.	For $k\in\{1,...,n\}$, the stationary density $p^k_{\infty}$   of the $k$-th normalized factor, when it is evolving in the  market time $\tau$,
       	is by the SDE \eqref{bardYtau} the  solution of the  stationary Fokker-Planck equation 
       	\begin{equation}
       		\frac{d}{dy}	 \left(p^k_\infty(y)y((\phi^k(y))^{-1}4\omega^k-1)a^k\right) -\frac{1}{2}\frac{d^2}{dy^2} \left(p^k_\infty(y)y^2\frac{4a^k}{\phi^k(y)}\right) =0,
       	\end{equation} as described, e.g.,in  Chapter 4 in \citeN{PlatenHe06}, which is a  second-order ODE. Its solution is given by the formula 
       	\BE \label{qjy}
       	p^{k}_\infty(y)=\frac{C_k\phi^k(y)}{4y^2 }\exp \left\{2\int_{1}^{y}\frac{\omega^k-\frac{\phi^k(u)}{4}}{u}du\right\}
       	\EE
       	for $ y\in(0,\infty)$ and some constant $C_k>0$. The latter ensures that $p^k_\infty$ is a probability density. \\
       	
       \noindent 2.	Under the  constraints \eqref{omega} and \eqref{barstalogmean*},   we minimize according to \eqref{Ipq} the surprisal \begin{equation}\mathcal{I}(p^k_\infty)=-\int_{0}^{\infty}\ln(p^k_\infty(y))p^k_{\infty}(y)dy\end{equation}    of the stationary probability density $p^{k}_\infty$ of the $k$-th independent normalized factor $Y^k_{\tau^k_.} $, $k\in\{1,...,n\}$. According to the formula \eqref{Ipq}, we minimize the  Lagrangian
       	\begin{equation*}
       		{\cal{L}}(p^{k}_\infty, \lambda_0,\lambda_1,\lambda_2)= 
       		-\int_{0}^{\infty} \ln(p^k_\infty (y))p^k_\infty (y)dy+\lambda_0 \left(\int_{0}^{\infty}p^{k}_\infty(y)dy-1\right)\end{equation*}\BE +\lambda_1\left(\int_{0}^{\infty}y p^{k}_\infty(y)dy-4\omega^k\right)
       	+\lambda_2\left(\int_{0}^{\infty}\ln(y)p^{k}_\infty(y)dy-\zeta^k\right),
       	\EE
       	where $\lambda_0, \lambda_1, \lambda_2$ are Lagrange multipliers. $ {\cal{L}}(p^{k}_\infty, \lambda_0,\lambda_1,\lambda_2)$ is minimized when its Fr\'echet derivative $\delta {\cal{L}}(p^{k}_\infty, \lambda_0,\lambda_1,\lambda_2)$, i.e., the first variation of ${\cal{L}}(p^{k}_\infty, \lambda_0,\lambda_1,\lambda_2)$ with respect to admissible variations of $p^{k}_\infty$, becomes zero. This implies for the surprisal-minimized stationary density $ p^{k}_\infty$ the equation
       	\BE
       	\delta {\cal{L}}( p^{k}_\infty, \lambda_0,\lambda_1,\lambda_2)=\int_{0}^{\infty}\left(-\ln( p^{k}_\infty(y))+\lambda_0+\lambda_1 y+\lambda_2 \ln(y)\right)\delta  p^{k}_\infty(y) dy=0.
       	\EE
       	The solution of the above first-order condition is the gamma density 
       	\BE \label{barpy}
       	p^{k}_\infty(y)=\exp \{\lambda_0+ \lambda_1 y+\lambda_2 \ln(y)\}
       	\EE
       	for $y\in (0,\infty)$ with the constraint \begin{equation}
       		\int_{0}^{\infty}\exp \{\lambda_0+ \lambda_1 y+\lambda_2 \ln(y)\}dy=1,
       	\end{equation}
       	and the Lagrange multipliers $\lambda_0, \lambda_1,\lambda_2$. It
       	has  $2(\lambda_2+1)$ degrees of freedom and  as parameters and constraints the averages \BE {\bf{E}}^{ p^{k}_\infty}(Y^k_{.})=\frac{\lambda_2+1}{-\lambda_1}=4\omega^k\EE %and 
       	and
       	$$ {\bf{E}}^{ p^{k}_\infty}(\ln(Y^k_{.}))=\zeta^k.$$
       	\noindent 3. The SDE for the   $k$-th normalized factor  is given  by \eqref{bardYtau}.  
       	Consequently, the stationary density $p^{k}_\infty(y)$ of the $k$-th normalized factor satisfies the Fokker-Planck equation  with the drift and diffusion coefficient functions of the SDE \eqref{bardYtau}. This yields  the stationary density $p^{k}_\infty(y)$ in the form  \eqref{qjy}. The latter must equal the above-identified gamma density. By setting  both expressions for the stationary density equal, respective conditions for the  function $\phi^k(y)$ emerge.\\
           	The Weierstrass Approximation Theorem states that a continuous function can be approximated on a bounded interval by a polynomial.  When using a polynomial for characterizing $\phi^k(y)$ and searching for a match of the stationary density \eqref{qjy} with the gamma density  \eqref{barpy}, then one finds by comparing the coefficients of the possible polynomials  that only the polynomial
       	\BE\label{psi1}
       	\phi^k(y)=y 
       	\EE
       	provides such a match. 
       	This yields for the $k$-th normalized factor process $Y^k_.$ the 
       	stationary density
       	\BE\label{qY1}
       	p^{k}_\infty(y)=\frac{C_k y }{4y^2 }\exp \left\{2\int_{1}^{y}\frac{\omega^k-\frac{u}{4} }{u}du\right\}=\frac{y^{2\omega^k-1}}{2^{2\omega^k}\Gamma(2\omega^k)} \exp\left\{-\frac{y}{2}\right\}
       	\EE
       	for $y>0$. The above density  is the gamma density with $d_k=4\omega^k$ degrees of freedom and mean $4\omega^k$. We assumed
       	the logarithmic average of the stationary density to equal   a constant $\zeta^k$, which emerges as
       	\BE\label{Hq'}
       	\zeta^k=  {\bf{E}}^{ p^{k}_\infty}(\ln(Y^k_{.}))=\ln\left(2\right)+  \psi(2\omega^k),
       	\EE
       	where the function $\psi(x)$ 
       	is the diagamma function;  see, e.g., \citeN{AbramowitzSt72} or \citeN{JohnsonKoBa95}.\\ 
       	\noindent 4.	The resulting CIR process is a stationary process where its initial value $Y^k_{\tau_{t_0}}=S^k_{\tau_{t_0}}$ is distributed according to its stationary density. 
       This yields with
       \begin{equation}\sqrt{\mathcal{I}}=\sum_{k=1}^{n}\frac{1}{n}\sqrt{\mathcal{I}(p^k_\infty)}
       \end{equation}
       the surprisal of the stationary joint density
       \begin{equation*}\mathcal{I}(p_\infty)=n\sum_{k=1}^{n}\frac{1}{n}\mathcal{I}(p^k_\infty)=n\sum_{k=1}^{n}\frac{1}{n}\left(\sqrt{\mathcal{I}(p^k_\infty)}\right)^2\end{equation*}\begin{equation}=n\left((\sqrt{\mathcal{I}})^2+\sum_{k=1}^{n}\frac{1}{n}\left(\sqrt{\mathcal{I}(p^k_\infty)}-\sqrt{\mathcal{I}}\right)^2\right).
       \end{equation}
       The latter is minimized for 
       \begin{equation}
       	\sqrt{\mathcal{I}(p^k_\infty)}=\sqrt{\mathcal{I}},
       \end{equation} 
       which means by \eqref{Nt=1} that
       \begin{equation}
       	\omega^k=\frac{1}{n}.
       \end{equation} 
       According to \citeN{GoingYo03}, the $k$-th normalized factor is a CIR process and, more generally, an SROU process that represents a unique strong solution of the SDE \eqref{dY21}, which completes   the proof of Theorem \ref{informationtheorem1}.
       \end{proof}
              
       	\textwidth14.5cm \textheight9in
       	       	\topmargin0pt
       	\renewcommand{\theequation}{\mbox{D.\arabic{equation}}}
       	\renewcommand{\thefigure}{D.\arabic{figure}}
       	\renewcommand{\thetable}{D.\arabic{table}}
       	\renewcommand{\footnoterule}{\rule{14.8cm}{0.3mm}\vspace{+1.0mm}}
       	\renewcommand{\baselinestretch}{1.0}
       	\pagestyle{plain}
       	\section*{Appendix D: Proof of Theorem \ref{stockGOP11}}\setcA \setcB
       	\begin{proof} \\
       	\noindent	1. Let us consider a market time $\tau\in\mathcal{T}$ when no normalized factor equals zero. By employing Theorem \ref{stockGOP} with \eqref{e.4.122a} and \eqref{beta}, 
       		it follows  for the benchmark the SDE
       		       	\begin{equation}\label{e.4.1222}
       		\frac{dS^{*}_\tau}{S^{*}_\tau}
       		= \hat\lambda_\tau d\tau +
       		\sum_{k=1}^{n}\frac{1}{n}\sqrt{\frac{4 a^k }{Y^k_{\tau^k_\tau}}}\left(\frac{1}{n}\sqrt{\frac{ 4a^k}{Y^k_{\tau^k_\tau}}} d\tau+d{ W}^k_\tau\right). \end{equation}
       	       	By using \eqref{Zt}, one can rewrite this SDE  in the form
       	\begin{equation} \label{e.4.119}
       		d  S^*_\tau= \hat\lambda_\tau S^*_\tau d\tau+
       		Z_\tau S^*_\tau d\tau
       			+		 \sqrt{Z_\tau}S^*_\tau d{ W}^*_\tau,
       		       		\end{equation}
       		where   the Brownian motion $W^*_\tau$ has the stochastic differential given in \eqref{W*t}
       		with initial value $W^*_{\tau_{t_0}}=0$. This  proves the SDE \eqref{e.4.1181} together with \eqref{Zt} and 
       		\eqref{W*t}.\\
       		2. Similarly as in \eqref{Yktayt1}, one can introduce the normalized benchmark $Y^*_{\tau^*_\tau}$ via the formula \eqref{Y*1}.
       		By application of the It\^o formula,  Equation \eqref{Y*1}, and \eqref{Zt} it follows that $Y^*_{\tau^*_\tau}$  satisfies the SDE
       		\begin{equation*}
       			\frac{d Y^*_{\tau^*_\tau}}{Y^*_{\tau^*_\tau}}=\hat\lambda_\tau d\tau+\left(Z_\tau-a^*_\tau\right)d\tau+\sqrt{Z_\tau}dW^*_\tau\end{equation*}
       		\begin{equation}=\hat\lambda_\tau d\tau+\left(\frac{4}{Y^*_{\tau^*_\tau}}-1\right)a^*_\tau d\tau+\sqrt{\frac{4a^*_\tau}{Y^*_{\tau^*_\tau}}}dW^*_\tau
       		\end{equation}
       		and	is forming an SROU  process of dimension four  satisfying the SDE \eqref{Y^*_t} that
       		evolves in the benchmark time 
       		\begin{equation}
       			\tau^*_\tau=\tau^*_{\tau_{t_0}}+\int_{\tau_{t_0}}^{\tau}a^*_sds
       		\end{equation}
       		with the  benchmark activity           
       		\begin{equation}
       			a^*_\tau=\frac{1}{4}Z_\tau Y^*_{\tau^*_\tau}.
       		\end{equation} 
       	
       	This proves the remaining statements of Theorem \ref{stockGOP11}.
       	\end{proof}\\

       	 \textwidth14.5cm \textheight9in
       	        	 \topmargin0pt
       	 \renewcommand{\theequation}{\mbox{E.\arabic{equation}}}
       	 \renewcommand{\thefigure}{E.\arabic{figure}}
       	 \renewcommand{\thetable}{E.\arabic{table}}
       	 \renewcommand{\footnoterule}{\rule{14.8cm}{0.3mm}\vspace{+1.0mm}}
       	 \renewcommand{\baselinestretch}{1.0}
       	 \pagestyle{plain}
       	 \section*{Appendix E: Proof of Theorem \ref{informationtheorem}}\setcA \setcB

       	\begin{proof} 
       	For   a surprisal-minimized  market one can write the Radon-Nikodym derivative $\Lambda_{Q_{S^*}}(\tau)$, given in \eqref{lambda1}, of the BN pricing measure $Q_{S^*}$   as the product
       	\begin{equation}\label{lambda}
       		\Lambda_{Q_{S^*}}(\tau)=\prod_{k=1}^{n}\Lambda_{Q_{S^*}}^k(\tau),
       	\end{equation}
       	of the independent  martingales
       	\begin{equation}\label{LambdaQS*a}
       		\Lambda_{Q_{S^*}}^k(\tau)	
       		=\exp\left\{-\frac{1}{2}\int_{\tau_{t_0}}^{\tau}\frac{\hat\lambda^2Y^k_{\tau^k_s}}{4a^k}ds-\int_{\tau_{t_0}}^{\tau}\sqrt{\frac{\hat\lambda^2Y^k_{\tau^k_s}}{4a^k}}dW^k_s\right\}
       	\end{equation}
       	for $k\in\{1,...,n\}$ and $\tau \in [\tau_{t_0},\infty)$. 	According to Equation \eqref{Ipq1},  \eqref{LambdaQS*a}, and \eqref{omega} we have to minimize the Kullback-Leibler divergence  
       		\begin{equation}I( p_{\infty},q^{}_\infty)=\frac{1}{2}	\sum_{k=1}^{n}{\bf{E}}^{p_\infty}\left(	\frac{\hat\lambda_\tau^2Y^k_{\tau^k_\tau}}{4a^k}\right)=\frac{1}{2n}	\sum_{k=1}^{n}	\frac{\hat\lambda_\tau^2}{a^k}
       			\rightarrow \min\end{equation}	
       		of $q_\infty$ from $p^{}_\infty$.       		
       		In \eqref{averactivity},  the average of the inverse of the square root of the activity in market time is set to
       		\begin{equation}\label{averactivity1}
       			\sum_{k=1}^{n}\frac{1}{n}\sqrt{\frac{1}{a^k}}=1.
       		\end{equation}
       		       			This allows us to write the Kullback-Leibler divergence in the form
       			\begin{equation*}
       				I( p_{\infty},q^{}_\infty)= 	\frac{\hat\lambda_\tau^2}{2n}\sum_{k=1}^{n}\frac{1}{a^k}= 	\frac{\hat\lambda_\tau^2}{2}\left(\left(\sum_{k=1}^{n}\frac{1}{n}\sqrt{\frac{1}{a^k}}\right)^2+\sum_{k=1}^{n}\frac{1}{n}\left(\sqrt{\frac{1}{a^k}}-1\right)^2\right)
       			\end{equation*} 
       			\begin{equation*}
       				= \frac{\hat\lambda_\tau^2}{2}\left(1+\sum_{k=1}^{n}\frac{1}{n}\left(\sqrt{\frac{1}{a^k}}-1\right)^2 \right).
       			\end{equation*}
       			This expression is minimized with respect to the  activities if all activities are equal with value
       			\begin{equation}
       				a^k=1,
       			\end{equation} which proves \eqref{act} and yields
       			\begin{equation}\label{GSS1}
       				I( p_{\infty},q^{}_\infty)
       				= \frac{\hat\lambda_\tau^2}{2}
       			\end{equation}
       			for $\tau\in[\tau_{t_0},\infty)$.	\\
       			Since the stationary probability densities $p_\infty$ and $q_\infty$ do not change over market time,
       			we obtain the net risk-adjusted return
       			\begin{equation}
       				\hat\lambda_\tau=\sqrt{2I( p_{\infty},q^{}_\infty)}=\hat\lambda
       			\end{equation}
       			as a constant. 
       			This  proves \eqref{infoI} and, therefore, Theorem \ref{informationtheorem}.    
       		\end{proof}

\bibliographystyle{chicago}
\bibliography{my}

\newpage

 \end{document}